\documentclass[10pt]{article}
\usepackage[english]{babel}
\usepackage[utf8]{inputenc}
\usepackage{geometry,color,framed,enumerate}                
\usepackage{amsmath,amsfonts,amsthm,amssymb,graphicx,bm,dsfont}
\usepackage{bbm}
\usepackage{epstopdf}
\usepackage{pgfplots}
\usepackage{caption,subcaption}
\usepgfplotslibrary{groupplots}

\newtheorem{theorem}{Theorem}[section]
\newtheorem{lemma}[theorem]{Lemma}

\newcommand{\R}{\mathbb{R}}

\newcommand{\seco}{{\prime\prime}}
\renewcommand{\sec}{\seco}
\newcommand{\eps}{\varepsilon}

\newcommand{\var}{\mathrm{Var}\,}
\newcommand{\ud}{\,d}
\def\G{\mathcal{G}}
\def\ds{\displaystyle}

\def\E{\mathbf{E}}

\newcommand{\M}{\mathcal{M}}
\newcommand{\La}{\mathcal{L}}
\newcommand{\C}{\mathcal{C}}
\newcommand{\abs}[1]{\left\lvert#1\right\rvert}
\newcommand{\norm}[1]{\left\lVert#1\right\rVert}

\newcommand{\p}{\partial}
\newcommand{\ind}{\ensuremath{\, \mathbbm{1}}}
\renewcommand{\tilde}{\widetilde}
\newcommand{\psii}{f}

\renewcommand{\L}{\mathbf{L}}
\renewcommand{\hat}{\widehat}

\renewcommand{\geq}{\geqslant}
\renewcommand{\le}{\leqslant}
\renewcommand{\ge}{\geqslant}
\renewcommand{\L}{\mathrm{L}}
\newcommand{\s}{\nu}
\renewcommand{\H}{I}
\newcommand{\h}{i}
\renewcommand{\phi}{\varphi}
\renewcommand{\P}{P}

\usepackage{tikz}
\usepackage{pgfplots} 
\pgfplotsset{compat=newest}
\usetikzlibrary[pgfplots.groupplots]
\usepgfplotslibrary{groupplots}
\usetikzlibrary{matrix,arrows,decorations.pathmorphing}
\definecolor{k4}{rgb}{0.8,0.8,0.8}
\definecolor{k3}{rgb}{0.6,0.6,0.6}
\definecolor{k2}{rgb}{0.4,0.4,0.4}
\definecolor{k1}{rgb}{0.2,0.2,0.2}

\newcommand\restr[2]{{
  \left.\kern-\nulldelimiterspace 
  #1 
  \vphantom{\big|} 
  \right|_{#2} 
  }}

\title{Existence of recombination-selection equilibria \\
for sexual populations} 
\author{Thibault Bourgeron, Vincent Calvez,  Jimmy Garnier, Thomas Lepoutre}
\date{March 15, 2017}

\begin{document}
\maketitle


\begin{abstract} We study a birth and death model for the adapatation of a sexual population to an environment. The population is structured by a phenotypical trait, and, possibly, an age variable. Recombination is modeled by Fisher's infinitesimal operator. We prove the existence of principal eigenelements for the corresponding eigenproblem. As the infinitesimal operator is $1$-homogeneous but nor linear nor monotone, the general Kre\u{\i}n-Rutman theory cannot be applied to this problem.
\end{abstract}


\medskip


\section{Introduction}

\subsection{The infinitesimal model}

We first recall the definition and some properties of the infinitesimal model, see \cite{burger, doebeli, metz} for biological motivation, and \cite{MR,DJMP,DJMR,LMP} for mathematical properties of this kind of models.
The infinitesimal model provides a simple and robust model describing inheritance of quantitative traits. It can be rigorously derived, see \cite{BEV}, from a discrete model with Mendelian inheritance, mutation and environmental noise, when the genetic component of the trait is purely additive (non-epistasis case) or when epistasis is not consistent in direction, as the number of underlying loci tends to infinity.

\medskip

In this model males and females are not distinguished. This makes senses, for instance, if they have the same distribution (\emph{e.g.} for hermaphrodites).
Mating is assumed to be random and uniform among the population.
The number of offsprings is assumed to be proportional to the density (of females).
Let us consider a density of individuals $f(z)$ structured by a continuous phenotypical trait $z \in \R$. Sexual reprodution involves two parents, having the traits $z^\prime$ and $z^{\prime\prime}$ who give birth to an offspring having the trait $z$.
In the infinitesimal model the genetic component of offspring traits follows a normal distribution centered the average of the parental traits $\frac{z^{\prime}+z^{\prime\prime}}{2}$, while this distribution has a variance that is independent of the parental values and which remains constant. The density of descendants is normalized to be proportional to the density of individuals $f$. This leads to consider the following mixing operator:
\begin{equation}\label{eq:def-G}
\G f(z):=\dfrac1{\int_\R f(x)\, dx} \, \iint_{\R^2} f(z^\prime) f(z^\sec)\, G\left(z - \dfrac{z^\prime + z^\sec}2\right) \ud z^\prime \ud z^\sec,
\end{equation}
where $G$ is the gaussian distribution with the given variance $V_{LE}/2$, where $V_{LE}$ is the variance at linkage equilibrium, see \cite{bulmer, TB, tufto, barfield, huisman}.

\medskip

As noticed in \cite{TB}, and used in \cite{MR} for numerical simulations, the operator $\G$ can be defined as a double convolution product:
\begin{equation}\label{eq:def-Gconv}
\G f (z) =\dfrac4{\int_\R f(x)\, \ud x} \, \left( G * f(2\cdot) * f(2\cdot)\right)(z),
\end{equation}
or equivalently through its Fourier transform, with the convention $\hat{f}(\xi) = \int_\R e^{-i x \xi} f(x) \ud x$: 
\begin{equation}\label{eq:def-GFourier}
\hat{\mathcal \G f} (\xi) = \dfrac1{\hat{f}(0)} \, \hat{G}(\xi) \hat{f}(\xi/2)^2.
\end{equation}

\medskip

Beyond the biological background, we do not assume that the probability density function $G$ is a gaussian function but we consider slightly more general cases.

\subsection{The recombination-selection model without age structure}

To model the adaptation of a population to an environment, let us consider a population density $f(t,z)$ where $t\ge 0$ is the time variable and $z$ the trait variable.
We assume the time variable $t$ to be continuous. This implies that overlapping generations exist that is more than one breeding generation is present at any time.
For the sake of simplicity we assume that the birth rate is constant and that selection is taken into account through the heterogeneous mortality rate $\mu(z)$.

\medskip

These assumptions lead us to consider the following recombination-selection equation:
\begin{equation} \label{eq:noage}
\p_t f(t,z) = - \mu(z) f(t,z) + \beta \, \G(f(t,\cdot))(z),
\end{equation}
where $\G$ is defined by \eqref{eq:def-G}.
The model \eqref{eq:noage} will be referred to as the homogeneous model.

\medskip

The death rate $\mu(z)$ is assumed to be a symmetric function in the trait variable $z$, so that the mortality is minimal at $z=0$.
We further assume that $\mu(0)=0$, without loss of generality.
The term $\mu f$ models selection: individuals having the optimal trait $z=0$ have the minimal mortality rate, so are more likely to be maintained in the population.
The second term $\beta \, \G(f)$ is the birth term and models inheritance of the traits by the sexual reproduction, at birth.


\medskip

For $\mu = \beta$ constant, one can check that the size of the population is constant and the mean phenotypical trait is conserved.
It has been shown, see \cite{bulmer, TB}, that \eqref{eq:noage} has a unique steady-state, that is the Gaussian distribution of variance $\sigma^2 = V_{LE}$. Furthermore, it can be shown that the distribution $f(t,\cdot)$ converges exponentially fast to the Gaussian distribution in the sense of Wasserstein distances, see \cite{MaR}.

\subsection{The age-structured recombination-selection model}

The antagonistic pleiotropy theory explains senescence by the selection of genes with positive effects at young ages even though they can be deleterious later in life, see \cite{williams}. This idea was formalized by Hamilton, \cite{hamilton}, and developped by Charlesworth, \emph{cf.} \cite{charlesworth1980, charlesworth1994, charlesworth2001}, and recently numerically studied in \cite{CR}. To model ageing of individuals we thus introduce an age variable $a \ge 0$ in Equation \eqref{eq:noage}, see \cite{perthame} for classical age-structured PDEs. This leads to consider the age-structured recombination-selection equation:
\begin{equation}\label{eq:age}\begin{cases}
\partial_t f(t,a,z) + \partial_a f(t,a,z) + \mu(a,z) f(t,a,z) = 0, \medskip\\
\displaystyle f(t,0,z) = \G\Big(\int_0^\infty\!\!\beta(a)f(t,a,\cdot)da\Big)(z),
\end{cases}
\end{equation} 
where $\G$ is defined by \eqref{eq:def-G}. The death rate $\mu(a,z)$ and the birth rate $\beta(a)$ can depend on the age variable $a$.
The first equation models that the age is increasing at speed one while the mortality is acting as for the homogeneous model \eqref{eq:noage}.
The second equation models the inheritance of the trait at birth (age $0$).

\medskip

\textbf{Examples.} For different models, the age-specific effect of mutations are considered to affect a single age class or a range of age class.
Different assumptions about the age-specific effect of mutations translate into different patterns of age-specific mortality, see \cite{abrams, charlesworth2001, baudisch, moorad, wachter}. In the previous age-continuous setting, some of these assumptions can be translated into considering death rates of the form:
\begin{equation} \label{examples}
\mu(a,z) =  \abs{z}^\alpha \delta_{a=a^*}, \quad \text{or} \quad  \abs{z}^\alpha \ind_{a\ge a^*},
\end{equation}
where $\alpha>0$ and $a^*>0$.

\subsection{The principal eigenelements}

As the operators $f \mapsto \beta f$ and $f\mapsto \mu \G(f)$ are $1$-homogeneous, special solutions to \eqref{eq:noage} of the form $f(t,z) = e^{\lambda t} F(z)$, resp. solutions to \eqref{eq:age} of the form $f(t,a,z) = e^{\lambda t} F(a,z)$, are characterized by pairs $(\lambda,F)$ solving the following spectral problems. The real number $\lambda$ is the rate at which the number of individuals grows $(\lambda \ge 0$) or decays ($\lambda \le 0$). 

\medskip

For the homogeneous model \eqref{eq:noage} the spectral problem is to find a pair $(\lambda, F)$ solution of:
\begin{equation} \label{eq:pb1}
\begin{cases}
\ds  \lambda F(z) + \mu(z) F(z) = \G \left(F \right)(z), \medskip\\
\ds \int_\R F(z) \ud z = 1.
\end{cases}
\end{equation}

\medskip

For the age-structured model \eqref{eq:age} the spectral problem is to find a pair $(\lambda, F)$ solution of:
\begin{equation} \label{eq:pb2}
\begin{cases}
\ds  \lambda F(a,z) + \partial_a F(a,z) + \mu(a,z) F(a,z) = 0, \medskip\\
\ds F(0,z)= \G\Big( \int_0^\infty{\beta(a)F(a,\cdot)\,da} \Big)(z), \medskip\\
\ds \int_\R F(0,z) \ud z = 1.
\end{cases}
\end{equation}

As Problem \eqref{eq:pb2} is $1$-homogeneous with respect to $F$, we chose to normalize the density of newborns to be $1$.

\bigskip

The operator $\G$ defined by \eqref{eq:def-G} is not linear, so that the classical Kre\u{\i}n-Rutman theory, see \cite{krein, dautraylions}, cannot be applied to this problem.
The operator $\G$ is positively $1$-homogeneous and compact on subsets of $\L^1(\R)$ with bounded variance, \emph{cf.} proof of Theorem \ref{th1}.
But the operator $\G$ is not monotone, because of the division by $\int_\R f \ud x$, so the generalization of Kre\u{\i}n-Rutman theory to $1$-homogeneous operators cannot be applied neither, see for instance \cite{mahadevan}.

\section{Existence of recombination-selection equilibria}

\subsection{The model without age structure, assumptions and statements}

{\bf Assumptions on $G$.} The probability density function $G$ is (essentially) positive, symmetric, uniformly continuous and it has a finite second moment.

\medskip

\noindent {\bf Assumptions on $\mu$.}
\begin{enumerate}
   \item The death rate $\mu(z)$ is non-negative, symmetric, non-decreasing for $z\ge 0$ and satisfies $\mu(0)~=~0$.
   \item We further assume integrability conditions at $0$ and $\infty$:
   \begin{equation} \label{hyp:th1-1}
\frac1{\mu} \not\in \L^1(-1,1),
\end{equation}
   \begin{equation} \label{hyp:th1-2}
\frac1{\mu} \in \L^1(1,\infty).
\end{equation}
\end{enumerate}

A typical example is $\mu(z) = \abs{z}^\alpha$ with $\alpha>1$. Assumption \eqref{hyp:th1-1} will be referred to as weak selection, in contrast to strong selection, \emph{cf.} \S \ref{strong}.
It can be rephrased as: $H\left(\restr{\mu}{(-1,1)}\right)$ is zero, where $H$ is the harmonic mean defined by:
\[
H(f) = \left(\int_\R \frac{dz}{f(z)}\right)^{-1}.
\]

A similar weak selection assumption in the case of asexual reproduction is assumed in \cite{martinroques}.

\begin{theorem}[weak selection] \label{th1} Let $G$ be a function satisfying the Assumptions on $G$. We consider the resulting operator $\G$ defined by \eqref{eq:def-G}. Let $\beta$ be a positive number and $\mu$ a function satisfying the Assumptions on $\mu$. Then, there exists a solution $(\lambda,F)$ to Problem \eqref{eq:pb1} with $\lambda >0$ and $F$ a non-negative $\L^1(\R)$ function. \end{theorem}

\medskip

\begin{figure} \label{fig1}
\centering

\includegraphics[width=15cm]{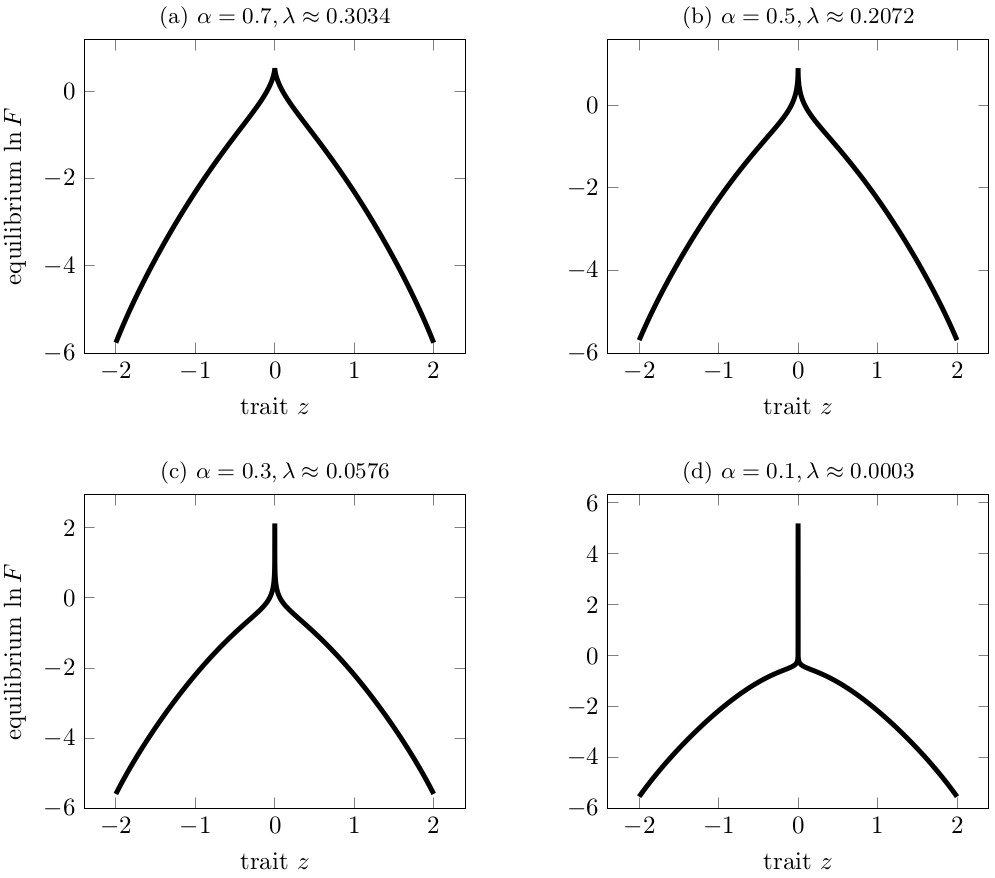}

\caption[]{The selection-mutation equilibrium $F(z)$ obtained from \eqref{eq:noage},
for $\mu(z) = \abs{z}^{\alpha} + \abs{z}^2$. It may have a Dirac mass at $z=0$ when the limiting eigenvalue $\lambda=0$ is reached. }
\end{figure}

Assumption \eqref{hyp:th1-1} can be relaxed. In this case it may happen that the limiting value $\lambda=0$ is reached and $F$ has a Dirac mass at the origin, see Figure 1. 
To include more singular behaviour at the origin, we consider the normed vector space:
\begin{equation} \label{eq:defD}
\E = \R \delta + \L^1(\R) = \left\{x \delta + f, x\in \R, f\in \L^{1}(\R) \right\},
\end{equation}
where $\delta$ is the Dirac measure at $\{z=0\}$. It is endowed with a product norm on $\R \times \L^1(\R)$.

\begin{theorem}[strong selection] \label{th3} Let $G$ be a function satisfying the Assumptions on $G$. We consider the resulting operator $\G$ defined by \eqref{eq:def-G}. Let $\beta$ be a positive number and $\mu$ a function satisfying the Assumption \eqref{hyp:th1-2}. Then, there exists a solution $(\lambda,F)$ to Problem \eqref{eq:pb1} with $\lambda \ge 0$ and $F \in \E$, defined by \eqref{eq:defD}, the singular and the regular parts of which are non-negative. \end{theorem}

For a convolution operator a similar result is obtained in \cite{coville2013}.

\medskip

These recombination-selection equilibria are built as fixed point of some operator linked to $\G$.
Assuming that $G$ and $\mu$ are symmetric allows us to build symmetric equilibria.
This operator is shown to stabilize the variance, thus $G$ is needed to have a finite second moment, \emph{cf.} Lemma \ref{lem:moments}. Also, the key Lemma \ref{lem:Var_sphi} shows that the multiplication by the mortality rate $\mu$ does not affect this property if $\mu$ is a symmetric and non-decreasing (for $z\ge0$) function. Lastly, the uniform continuity of $G$ and the integrability of $1/\mu$ at infinity are needed to establish the compacity of this operator, \emph{cf.} Proofs of the Theorems.

\subsection{The age-structured model, assumptions and statements}

{\bf Assumptions on $G$.} The probability density function $G$ is symmetric
and it has a finite second moment.

\medskip

\noindent {\bf Assumptions on $\beta$ and $\mu$.}
\begin{enumerate}
   \item The reproduction rate $\beta$ and the death rate $\mu(a,z)$ are non-negative. 
   \item The death rate $z\mapsto \mu(a,z)$ is symmetric, non-decreasing for $z\ge 0$ and satisfies $\mu(a,0)=0$.
   \item Our last and main assumption links the reproduction rate $\beta$ and the death rate $\mu$.
There exists $\lambda_0 \in \R$ such that:
\begin{equation} \label{hyp:th2}
\forall z \in \R \quad 1 \le \int_0^{\infty} \beta(a) e^{-\lambda_0 a} e^{-\int_0^a \mu(a^\prime,z) \ud a^\prime} \ud a \le \int_0^{\infty} \beta(a) e^{-\lambda_0 a} \ud a < \infty.
\end{equation}
\end{enumerate}

Assumption \eqref{hyp:th2} is a balance assumption between birth and death rates which is uniform in the trait variable. As the weak selection assumption \eqref{hyp:th1-1} for the homogeneous case, Assumption \eqref{hyp:th2} avoids concentration at $z=0$. For the death rates given by \eqref{examples} and a constant birth rate $\beta$, Assumption \eqref{hyp:th2} is satisfied as soon as $a^*>0$. 
From a technical viewpoint, Assumption \eqref{hyp:th2} ensures the existence of eigenvalues bounded from below by $\lambda_0$, \emph{cf.} Lemma \ref{lem:Laplace}.

\begin{theorem} \label{th2} Let $G$ be a function satisfying the Assumptions on $G$. We consider the resulting operator $\G$ defined by \eqref{eq:def-G}. Let $\beta(a)$ and $\mu(a,z)$ be functions satisfying the Assumptions on $\beta$ and $\mu$. Then, there exists a solution $(\lambda,F)$ to system \eqref{eq:pb2} with $\lambda \ge \lambda_0$. \end{theorem}

\section{Proofs of the existence theorems}

We first prove Theorem \ref{th1} because some estimates are easier and its proof contains all the steps necessary for the proof of Theorem \ref{th2}. Namely, the dispersion relation gives implicitly the principal eigenvalue corresponding to a principal eigenvector. This allows us to define a mixing operator, the fixed points of which are principal eigenelements. Proving the theorem comes down to build a convex set conserved by this operator in order to apply the Schauder fixed point theorem.

\subsection{Confinement and technical lemmas}

The proofs of these theorems rely on the variance stabilizing property of the operator $\G$, \emph{i.e.} $\G$ has a confinement effect in the trait variable.
The proofs do not use the confinement effect of the selection operator $f\mapsto \mu f$. Lemma \ref{lem:moments} describes how the infinitesimal operator acts on the first moments. Lemma \ref{lem:Var_sphi}, which is the key lemma, implies that the multiplication by a symmetric, non-increasing function reduces the variance. The third and last lemma is a short technical lemma concerning the Laplace transform. It is only used in the proof of Theorem \ref{th2}. 

\begin{lemma} \label{lem:moments} (first moments of the operator $\G$)
\begin{enumerate}
\item Let $\psii \in \L^1(\R)$ with $\int_\R \psii(z) \ud z \ne 0$. Then:
\[\int_\R \G \psii(z) \ud z = \int_\R \psii(z) \ud z, \quad \text{and} \quad \int_\R z \G \psii(z) \ud z = \int_\R z \psii(z) \ud z.\]
\item Let $\psii$ be an $\L^1(\R)$ function satisfying:
$\int_\R \psii(z) \ud z =1$.
The following equality holds: \[ \int_\R \left( \G \psii(z) - \psii(z)\right) \, z^2 \, \ud z = \var G - \frac12 \var f,\]
where $\var \phi = \int_\R z^2 \phi(z) \, dz - \left(\int_\R z \phi(z) \, dz\right)^2$.
In addition if $\int_\R z f(z) \ud z = 0$, one obtains:
\[\var\G(\psii) = \var G + \frac12 \var \psii.\]
As a consequence, if $\var \psii \le 2 \var G$ then $\var \G(\psii) \le 2 \var G$.
\end{enumerate}
\end{lemma}

\begin{proof} 
For any function $\phi$, from \eqref{eq:def-G}, using $\int_\R G(z) \, dz = 1$, one obtains:
\begin{equation} \label{test}
\begin{split}
& \displaystyle\int_{\R} \left( \G \psii(z) - \psii(z) \right) \phi(z) \, \ud z \\
& = \dfrac1{\int_\R \psii(x)\, dx} \,  \displaystyle\iiint_{\R^3} \psii(z^\prime) \, \psii(z^\sec) \, G(z) \, \left[ \phi\left(z + \frac{z^\prime + z^\sec}2\right) - \frac12 \left( \phi(z^\prime) + \phi(z^\sec) \right) \right] \ud z^\prime \ud z^\sec \ud z.
\end{split}
\end{equation}

\begin{enumerate}
\item The first point is obtained for $\phi(z) = 1$, and $\phi(z) = z$ using $\int_\R z G(z) \ud z = 0$.
\item Again, for $\phi(z) = z^2$, equality \eqref{test} leads to the result. After integrating the equality:
\[
\left( z + \dfrac{z^\prime + z^\sec}2 \right)^2 - \frac12 \left( z^{\prime2} + z^{\sec2} \right)
= z^2 - \frac{1}{4} z^{\prime 2} - \frac{1}{4} z^{\sec 2} + z z^\prime + z z^\sec + \frac{1}{2} z^\prime z^\sec,
\]
one gets:
\begin{eqnarray*}
& \displaystyle\int_\R \left(\G\psii(z)-\psii(z)\right) \, z^2 \ud z \\
&= \displaystyle\int_\R z^2 G \ud z \left( \displaystyle\int_\R \psii \ud z \right)
- 2 \frac14 \displaystyle\int_\R z^2 \psii \ud z \displaystyle\int_\R G \ud z   + \, 2 \displaystyle\int_\R z G \ud z \displaystyle\int_\R z \psii \ud z + \frac12 \frac{\left( \displaystyle\int_\R z \psii \ud z \right)^2}{\displaystyle\int_\R \psii \ud z} \displaystyle\int_\R G \ud z \\
&= \var G - \displaystyle\frac12 \var f.
\end{eqnarray*}
\end{enumerate}
\end{proof}

\begin{lemma}\label{lem:Var_sphi}
Let $\nu$ be a measurable, non-negative, symmetric on $\R$, non-increasing on $[0,+\infty)$ function.
Let $W$ be a $W^{1,1}_{loc}$, non-negative, symmetric on $\R$, non-decreasing on $[0,+\infty)$ function.
The quadratic form 
\[
Q(\phi)= \int_\R \nu \phi W \ud x \int_\R \phi \ud x  - \int_\R \phi W \ud x \int_\R \nu \phi \ud x
\]
is non-positive on the set:
\[
\mathcal D = \left\{\phi \ge 0: \phi, \nu\phi \in \L^1(\R,\ud x) \cap \L^1(\R, W \ud x)\right\}.
\]
\end{lemma}

\begin{proof}
The form $Q$ is linear with respect to $W$ and vanishes if $W$ is a constant function.
Thus considering $W - W(0)$ we can assume that $W(0)=0$.

\medskip

Let us define the bilinear symmetric form:
\[
\begin{split}
B(\phi,\psi) = \frac12
\Big(\int_\R \nu \psi W \ud x \int_\R \phi \ud x 
+  \int_\R \nu \phi W \ud x \int_\R \psi \ud x\\
- \int_\R \phi W \ud x \int_\R \nu \psi \ud x
-\int_\R \psi W \ud x \int_\R \nu \phi \ud x\Big),
\end{split}
\]
such that: $Q(\phi)=B(\phi,\phi)$. As: $B(\phi,\psi)=0$, if $\phi$ is odd, considering the decomposition of a function $\phi = \phi_e+\phi_o$ as the sum of its even  and its odd components, leads to:
\[ Q(\phi) = Q(\phi_e) + B(\phi_o,\phi_o) + 2 B(\phi_e,\phi_o) = Q(\phi_e).\]

Thus the functions in $\mathcal D$ can be assumed to be even. If $\phi$ is even, then:
\begin{equation} \label{eq:Qpaire}
\frac14 Q(\phi)= \int_0^\infty \nu \phi W \ud x \int_0^\infty \phi \ud x  - \int_0^\infty \phi W \ud x \int_0^\infty \nu \phi \ud x.
\end{equation}

\medskip

For a non-negative function $\phi:[0,\infty)\to \R$ which belongs to $\L^1(\R)$ let us denote: 
\[
F_\phi(x) = \int_x^\infty \phi(z) \ud z. 
\]

For $\phi \in \L^1(W \ud x)$, as $W$ is non-decreasing, in the limit $x\to \infty$, we have: $F_\phi(x) W(x) \le \int_x^\infty \phi W \ud z \to 0$.
By integration by parts, using $W(0)=0$, we get:
\[
 \int_0^\infty \phi W \ud x = \int_0^\infty F_{\phi} \, \frac{\ud W}{\ud x} \ud x.
\]
Thus, \eqref{eq:Qpaire} can be changed to:
\[ 
\frac14 Q(\phi) = \int_0^\infty \left(A F_{\nu \phi}(x) - B F_{\phi}(x)\right) \frac{\ud W}{\ud x} \ud x, \qquad \text{with} \quad A = \int_0^\infty \phi \ud x, B = \int_0^\infty \phi \nu \ud x.
\]

As $\frac{\ud W}{\ud x} \ge 0$, it amounts to show that $I = A F_{\nu \phi} - B F_{\phi}$ is non-positive.
The derivative of $-I$ is: \[J := - \frac{\ud I}{\ud x} = (A\nu-B) \phi. \]
The function $J$ satisfies $\int_0^\infty J(z) \ud z = I(0)-I(\infty) = 0$ so that the function $A\nu-B$ changes of sign.
As $A\nu-B$ is non-increasing there exists $x_0 \in [0,\infty)$ such that: $J(x)\ge 0$ for $x\le x_0$ and $J(x)\le 0$ for $x\ge x_0$.
Hence, $I$ is non-increasing on $[0,x_0]$ and non-decreasing on $[x_0,\infty)$, hence non-positive because $I(0)=I(\infty)=0$.
\end{proof}

\begin{lemma} \label{lem:Laplace} Let $\phi:\R^+\to\R$ be a measurable function, which is non-negative and non-zero.
For $p \in \R$ let us define: $\La(p) = \int_0^\infty e^{-p z} \phi(z) \ud z \in [0,\infty]$.
If there exists $\lambda_0 \in \R$ such that $1\le \La(\lambda_0)<\infty$, then there exists a unique $p \in \R$ such that $\La(p)=1$. 
\end{lemma}

\begin{proof} The function $\La(p)$ is finite for $p \ge \lambda_0$. Let us define the abscissa of convergence $p_0 = \inf \{ p \in \R \mid  \La(p) < \infty\} \in [-\infty,\lambda_0]$. Clearly, the function $\La$ is decreasing, analytic on $(p_0,\infty)$ and: $\lim_{p\to +\infty} \La(p) = 0$. Therefore there is at most one $p \in \R$ such that $\La(p)=1$. The existence of such an element is equivalent to $\La(p_0)\ge 1$. 
\end{proof}


\subsection{The recombination-selection equation without age structure, weak selection case}

\begin{proof}[Proof of theorem \ref{th1}] 
\begin{enumerate}
 \item For $\lambda> 0$ and $z\in\R$, or $\lambda\ge 0$ and $z\ne 0$, let us denote:
\begin{equation}\label{eq:nu1}
 \s(\lambda,z) = \frac{\beta}{\lambda + \mu(z)}.
\end{equation}

A solution $(\lambda,F)$ to \eqref{eq:pb1}, with $\lambda >0$, can be expressed as a fixed point using $\s$: \[ F = \s(\lambda,\cdot) \, \G(F). \]

Let us consider the following convex subset of $\L^1(\R)$:
\begin{equation} \label{eq:defC}
\mathcal{C}=\left\{ F\in \L^1(\R) \, \left\vert\,  F\geq 0, \int_\R F(z) \ud z =1,  F(-z)=F(z), \int_\R \abs{z}^2 F(z) \ud z \le 2 \var G \right.\right\}.
\end{equation}

The set $\mathcal{C}$ is a nonempty, bounded, closed, convex subset of $\L^1(\R)$. Regularization properties of the operator $\G$, defined as a double convolution \eqref{eq:def-Gconv}, makes the application of the Schauder fixed point theorem possible to the set $\mathcal{C}$ and the operator $\M$ defined by \eqref{eq:def-M1}. The construction of $\M$ is based on the following Claim.

\item  \textbf{Claim.} For any function $F$ in $\C$ there exists a unique real number $\lambda$ such that:
\begin{equation} \label{eq:deflambda}
 \int_\R \s(\lambda,z) \, \G(F)(z) \ud z = 1.
\end{equation}
This defines a function: $F \mapsto \lambda(F)$. The function $\lambda$ is bounded from below by a constant $\eps > 0$ which depends only on $G$ and $\mu$.

\medskip

The function $(0,\infty) \ni \lambda \mapsto \int_\R \s(\lambda,z) \, \G(F)(z) \ud z \in (0,\infty)$ is decreasing and tends to $0$ at infinity. So, it is enough to prove that its limit at $0$ is infinite.
Uniformly for $F$ in $\C$:
\begin{equation}\label{eq:tails}
\int_{\abs{z}>r} F(z) \ud z \le \frac{1}{r^2} \int_{\abs{z}>r} \abs{z}^2 F(z) \ud z \le \frac{1}{r^2} 2 \var G \to 0,
\end{equation}
so that there exists $r>0$ satisfying: $\int_{\abs{z}\le r} F \ud z \ge \frac12$.
By definition of $\G$, \eqref{eq:def-G}, for all $F$ in $\C$ and $z$ such that $\abs{z} \le r$, we obtain:
\begin{eqnarray*}
\G\big(F\big)(z) &\ge& \iint_{(z^\prime,z^\sec) \in [-r,r]^2} F(z^\prime) F(z^\sec)\, G\left(z - \dfrac{z^\prime + z^\sec}2\right) \ud z^\prime \ud z^\sec \\
&\ge& \left(\inf_{\abs{z} \le 2r} G\right) \left(\int_{\abs{z} \le r} F(z) \ud z\right)^2 \ge \left(\inf_{\abs{z} \le 2r} G\right) \times \frac14.
\end{eqnarray*}
So, we get:
\[
\int_\R  \s(\lambda,z) \, \G(F)(z) \ud z   \ge \int_{\abs{z}\le r} \s(\lambda,z) \, \G(F)(z) \ud z
\ge \frac14 \left(\inf_{\abs{z} \le 2r} G\right) \int_{\abs{z}\le r} \frac{\beta}{\lambda + \mu(z)} \ud z.
\]
Thanks to Fatou's lemma and Assumption \eqref{hyp:th1-1} the last quantity tends to $+\infty$ as $\lambda$ goes to $0^+$. In particular there exists $\eps >0$ such that:
\begin{equation} \label{eq:eps}
 \frac14 \left(\inf_{\abs{z} \le 2r} G\right) \int_{\abs{z}\le r} \frac{\beta}{\eps + \mu(z)} \ud z \ge 1,
\end{equation}
and we get the uniform lower bound for $\lambda$: $\lambda(F) \ge \eps$.
\item We define the mixing operator:
\begin{equation} \label{eq:def-M1}
\M(F) = \s(\lambda(F),\cdot) \, \G(F).
\end{equation}
\end{enumerate}

Proving Theorem \ref{th1} is equivalent to show that the operator $\M$ has a fixed point.

\begin{itemize}
\item The set $\C$ is preserved by $\M$. Let $F\in \C$. Clearly from \eqref{eq:def-G}, $\M F\ge 0$ and
thanks to the definition of $\lambda(F)$, $\int_\R \M F(z) \ud z =1$.
The equality $\M(F)(-z)=\M(F)(z)$ is a consequence of the symmetry of the death rate $\mu$ and the probability density function $G$.
The last point is a consequence of Lemmas \ref{lem:moments}, \ref{lem:Var_sphi} with $\nu = \s(\lambda,\cdot)$ and $W(z)=z^2$:
\[
\var \M(F) \le \var \G(F) = \var G + \frac12 \var  F \le 2 \var G.
\]

\item The set $\M(\mathcal{C})$ is relatively compact thanks to the Fréchet–Kolmogorov theorem.
As $\M(\C) \subset \C$, tails are uniformly bounded thanks to \eqref{eq:tails}.
Let $\tau_h$ be the translation by $h$ defined by: $\tau_h F = F(\cdot-h)$.
From the definition of $\M$, \eqref{eq:def-M1}, we get:
\begin{equation} \label{eq:t1}
\begin{split}
\norm{\tau_h (\M F) - \M F}_{\L^1(\R)}  &= \norm{\left(\tau_h (\s(\lambda(F),\cdot)) - \s(\lambda(F),\cdot)\right) \tau_h \left(\G F\right)
+ \s(\lambda(F),\cdot) \left(\tau_h \left(\G F\right) - \G F  \right)}_{\L^1(\R)}, \\
& \le \norm{\tau_h (\s(\lambda(F),\cdot)) - \s(\lambda(F),\cdot)}_{\L^1(\R)} \norm{\G F}_{\L^\infty(\R)}
+ \norm{\s(\lambda(F),\cdot)}_{\L^1(\R)} \norm{\tau_h \left(\G F\right) - \G F}_{\L^\infty(\R)}
\end{split}
\end{equation}

As $\lambda(F) \ge \eps$, we get:
\begin{equation} \label{eq:t1}
\begin{split}
\abs{\tau_h \s(\lambda(F),\cdot) - \s(\lambda(F),\cdot)} &= \frac{\abs{\mu - \tau_h \mu}}{(\lambda(F) + \tau_h \mu) (\lambda(F) + \mu)} \\
&\le \frac{\abs{\mu - \tau_h \mu}}{(\eps + \tau_h \mu) (\eps + \mu)} = \abs{\tau_h \s(\eps,\cdot) - \s(\eps,\cdot)}.
\end{split}
\end{equation}

In addition, by definition \eqref{eq:def-Gconv}, we have:
\begin{equation} \label{eq:t3}
\tau_h \left(\G F\right) - \G F = 4 (\tau_h G - G) * F(2 \cdot) * F(2\cdot).
\end{equation}

Using Young's inequality we get the uniform upper bound:
\[
\norm{\tau_h (\M F) - \M F}_{\L^1(\R)} \le \norm{\tau_h \s(\eps,\cdot) - \s(\eps,\cdot)}_{\L^1(\R)} \norm{G}_{\L^\infty(\R)}
+ \norm{\s(\eps,\cdot)}_{\L^1(\R)} \norm{\tau_h G - G}_{\L^\infty(\R)},
\]
which goes to $0$ as $h$ tends to $0$ because: $\s(\eps,\cdot)$ is in $\L^1(\R)$, see Assumption \eqref{hyp:th1-2}, and $G$ is uniformly continuous.

\item The operator $\M:\mathcal{C}\to\mathcal{C}$ is continuous.
Let $F_1, F_2 \in \C$, and let us denote: $\lambda_k = \lambda(F_k)$, $\s_k=\s(\lambda_k,\cdot)$, for $k=1,2$.
Using \eqref{eq:eps}, we have $\lambda_k \ge \eps$, and: $\norm{\s_k}_{\L^1(\R)} \le \norm{\s(\eps,\cdot)}_{\L^1(\R)} < \infty$. 

Writing:
\begin{equation} \label{eq:cont}
\begin{split}
\norm{\M(F_1)-\M(F_2)}_{\L^1(\R)} &= 4 \norm{(\s_1 - \s_2)\, \G(F_1) + \s_2\, (\G(F_1) - \G(F_2))}_{\L^1(\R)} \\
&\le 4 \norm{\s_1 - \s_2}_{\L^1(\R)} \norm{G}_{\L^\infty(\R)} + 4 \norm{\s(\eps,\cdot)}_{\L^1(\R)} \norm{\G(F_1) - \G(F_2)}_{\L^{\infty}(\R)}.
\end{split}
\end{equation}
it remains to show that the maps $\C \ni F \mapsto \mathcal \G(F) \in \L^{\infty}(\R)$ and  $\C \ni F \mapsto \s \in \L^1(\R)$ are continuous.

\medskip

Thanks to definition \eqref{eq:def-Gconv}, we have:
\begin{equation} \label{eq:a2b2}
\G(F_1) - \G(F_2) = 4\, G * (F_1 - F_2)(2\cdot) * (F_1 + F_2)(2\cdot).
\end{equation}
This implies that the first map is continuous thanks to Young's inequality.

\medskip

For the continuity of the second map we prove, as an intermediate stage, that $\C \ni F \mapsto \lambda \in \R$ is continuous. To do so let us fix $F_1$, define $\lambda = \lambda_2-\lambda_1$, and:
\begin{equation} \label{eq:I}
\H(\lambda) := \abs{\int_{\R} (\s_1 - \s_2) \, \G(F_1) \ud z}.
\end{equation}

\medskip

By definitions of the $\s_k, F_k$: $\int_\R \s_1 \G(F_1) \ud z = 1 = \int_\R \s_2 \G(F_2) \ud z$, and:
\begin{equation} \label{eq:maj}
\H(\lambda) = \abs{\int_{\R} (\s_1 - \s_2) \, \G(F_1) \ud z}
= \abs{ \int_\R \s_2 (\G(F_2)-\G(F_1)) \ud z}
\le \norm{\s(\eps,\cdot)}_{\L^1(\R)} \norm{\G(F_2)-\G(F_1)}_{\L^\infty(\R)}
\end{equation}
tends to $0$ as $\norm{F_1-F_2}_{\L^1(\R)}$ goes to $0$. 

We have:
\begin{equation} \label{eq:defIlambda}
\H(\lambda) = \abs{\int_\R (\s_1-\s_2) \, \G(F_1) \ud z} = \int_{\R} \frac{\abs{\lambda}}{(\lambda_1+\mu)(\lambda_1+\mu+\lambda)} \, \G(F_1) \ud z,
\end{equation}
where the last equality holds because the integrand has a constant sign.

It is easily checked from \eqref{eq:defIlambda} that the function $\H$ satisfies: $\H(0)=0$, it is continuous, strictly decreasing for $\lambda \le 0$ and strictly increasing for $\lambda \ge 0$. In particular: as $\H(\lambda)$ goes to $0$, $\abs{\lambda}$ tends to $0$.

\medskip

From \eqref{eq:maj}, we have: $\abs{\lambda}=\abs{\lambda_2-\lambda_1}$ goes to $0$.
This allows us to conclude thanks to Lebesgue's dominated convergence theorem and Assumption \eqref{hyp:th1-2}, because:
\begin{equation} \label{eq:contlambdanu}
\s_1 - \s_2 =(\lambda_1+\mu)^{-1}-(\lambda_2+\mu)^{-1}.
\end{equation}
\end{itemize}
\end{proof}

\subsection{The recombination-selection equation without age structure, strong selection case} \label{strong}

\begin{proof}[Proof of theorem \ref{th3}] 
\begin{enumerate}
 \item 

A solution $(\lambda,F)$ to \eqref{eq:pb1} can be expressed as a fixed point using $\s$ defined by \eqref{eq:nu1}: $F = \s(\lambda,\cdot) \, \G(F)$. Indeed the pair $(0, x \delta + f) \in \E$, where $\E$ is defined by \eqref{eq:defD}, is a solution to \eqref{eq:pb1} if and only if $\mu f = \G(x\delta +f)$ because $\mu(0)=0$.


\item Under the weak selection assumption \eqref{hyp:th1-1}, equation \eqref{eq:deflambda} for $\lambda\ge 0$, has a solution for every $F \in \C \subset \L^1(\R)$ because: $\int_\R \nu(0,z) \, \G(F)(z) \, dz = +\infty$. If assumption \eqref{hyp:th1-1} does not hold true, it may hapen that $\int_\R \nu(0,z) \, \G(F)(z) \, dz<1$ for some $F \in \L^1(\R)$. Thus equation \eqref{eq:deflambda} may have no solution. In such a case $\lambda=0$ and $F$ may have a Dirac part at the origin. This Dirac mass compensates the lack of mass of the regular part.

\medskip

First, this leads to consider the following convex subset of $\E$:
\begin{equation} \label{eq:defC}
\tilde{\mathcal{C}}=\left\{ m \delta + F \in \E \, \left\vert\,  m\ge 0, F\geq 0, m+\int_\R F(z) \ud z =1,  F(-z)=F(z), \int_\R \abs{z}^2 F(z) \ud z \le 2 \var G \right.\right\}.
\end{equation}

For $x\delta + f$ in $\E$ such that $x + \int_\R f \ud z\ne 0$, we have:
\begin{equation} \label{eq:GE}
\G(x\delta + f) = \frac1{x + \int_\R f \ud z} \left(x^2 \, G + 4 x \, G * f(2\cdot) + 4 \, G * f(2\cdot) * f(2\cdot) \right).
\end{equation}
Thus, $\G(x\delta + f) $ is an $\L^{1}(\R) \cap \L^{\infty}(\R)$ function. 

\medskip

Second, for $F$ in $\tilde\C$, we define the mixing operator $\tilde\M$ by:
\begin{equation} \label{eq:def-M3}
\tilde\M F =
\begin{cases}
\nu(\lambda,\cdot) \, \G(F) & \text{if } m=1- \int_{\R} \nu(0,\cdot) \, \G (F) \ud z < 0, \\
\nu(0,\cdot) \, \G(F) + m \delta & \text{if } m\ge 0, \\
\end{cases}
\end{equation}
where, in the first case, $\lambda>0$ is still defined by equation \eqref{eq:deflambda}, while in the second case, $\lambda=0$.

As a by-product of this definition, we have $m \times \lambda =0$, so that $m$ acts as a degree of freedom when $\lambda = 0$.


\end{enumerate}

\medskip

The Schauder fixed point theorem is applied to the operator $\tilde \M$ on the convex set $\tilde \C$, to prove Theorem \ref{th3}. Let us denote $p_1,p_2$ the canonical projections:
\[
p_1:\E\to \R, \quad p_2:\E\to \L^{1}(\R).
\]

\begin{itemize}
\item The set $\tilde{\mathcal{C}}$ is a nonempty, bounded, closed, convex subset of $\E$.
It is preserved by $\tilde \M$ thanks to definition \eqref{eq:def-M3}.
\item The set $p_1(\tilde{\M}(\tilde{\mathcal{C}})) \subset [0,1]$ is bounded and $p_2(\tilde{\M}(\tilde{\mathcal{C}})) \subset p_2(\tilde{\mathcal{C}})$ is relatively compact thanks to the Fr\'echet–Kolmogorov theorem.
Thanks to \eqref{eq:tails}, the tails of the elements of $p_2(\tilde \M(\tilde \C))$ are uniformly bounded. As in the previous case, from the definition of $\tilde \M$, \eqref{eq:def-M3}, we get:
\begin{equation*}
\begin{split}
 \norm{\tau_h p_2 (\tilde \M F) - p_2 (\tilde \M F)}_{\L^1(\R)} = \norm{\tau_h \nu(\lambda,\cdot) \, \G F - \nu(\lambda,\cdot) \, \G F}_{\L^1(\R)}\\
 \le \norm{\tau_h \s(0,\cdot) - \s(0,\cdot)}_{\L^1(\R)} \norm{G}_{\L^\infty(\R)}
+ \norm{\s(0,\cdot)}_{\L^1(\R)} \norm{\tau_h G - G}_{\L^\infty(\R)},
\end{split}
\end{equation*}
because $\lambda \ge 0$. This expression goes to $0$ as $h$ tends to $0$ because: $\s(0,\cdot)$ is in $\L^1(\R)$, see Assumption \eqref{hyp:th1-2}, and $G$ is uniformly continuous.

\item The operator $\tilde\M:\tilde\C\to\tilde\C$ is continuous.
Let $F_1, F_2 \in \tilde\C$, and let us denote: $\lambda_k = \lambda(F_k)$, $\s_k=\s(\lambda_k,\cdot)$, for $k=1,2$. We have: $\lambda_k\ge 0$ and $\norm{\nu_k}_{\L^1(\R)}\le \norm{\nu(0,\cdot)}_{\L^1(\R)}<\infty$.

First, the map $\tilde \C \ni F \mapsto \G (F) \in \L^{\infty}(\R)$ is continuous as a consequence of \eqref{eq:GE}, \eqref{eq:a2b2} and Young's inequality. Then, definition \eqref{eq:def-M3} yields:
\begin{eqnarray*}
\abs{p_1(\tilde \M F_1) - p_1(\tilde \M F_2)} & = & \abs{\left(1- \norm{\nu(0,\cdot) \, \G F_1}_{\L^1(\R)}\right)^+ - \left(1- \norm{\nu(0,\cdot) \, \G F_2}_{\L^1(\R)} \right)^+} \\
&\le& \abs{ \norm{\nu(0,\cdot) \, \G F_1}_{\L^1(\R)} - \norm{\nu(0,\cdot) \, \G F_2}_{\L^1(\R)} } \\
&\le& \norm{\nu(0,\cdot) \, (\G F_2 - \G F_1)}_{\L^1(\R)} \\
&\le& \norm{\nu(0,\cdot)}_{\L^1(\R)} \norm{\G F_2 - \G F_1}_{\L^\infty(\R)},
\end{eqnarray*}
thus, $p_1 \circ \tilde \M : \tilde\C\to \R$ is continuous.

Thanks to \eqref{eq:cont}, proving the continuity of $p_2\circ \tilde \M$ comes down to check the continuity of the maps $\tilde \C \ni F \mapsto \G (F) \in \L^{\infty}(\R)$ and $\tilde \C \ni F \mapsto \s \in \L^1(\R)$. The continuity of the first map has just been proved. Using \eqref{eq:contlambdanu}, the continuity of the second map is a consequence of the continuity of $\tilde\C \ni F \mapsto \lambda \in \R$, which is obtained using \eqref{eq:defIlambda}.

%

\end{itemize}
\end{proof}

\subsection{The age-structured recombination-selection equation}

The proof of Theorem \ref{th2} follows the same steps. The differences are that a first reduction step is needed but no lower bound on the eigenvalue has to be proved because it is provided by Assumption \eqref{hyp:th2}. 


\begin{proof}[Proof of theorem \ref{th2}] 
\begin{enumerate}
\item A solution to equation \eqref{eq:pb2} is determined by the phenotypical distribution at age $0$:
\begin{equation} \label{eq:defP}
\P=F(0,\cdot),
\end{equation}
through: $F(a,z) = F(0,z)  s(a,z) e^{-\lambda a}$, where: 
\begin{equation}
s(a,z)=e^{-\int_0^a \mu(a^\prime,z) \ud a^\prime},
\end{equation}
is the survival function.

Furthermore the function $\P$ can be expressed as a fixed point: $\P = \G\left( \s(\lambda,\cdot)\,  \P \right),$ where:
\begin{equation}\label{eq:nu2}
\s(\lambda,z) = \int_0^\infty\!\!\beta(a) s(a,z) e^{-\lambda a} \ud a.
\end{equation}



\item \textbf{Claim.} For any non-negative function $\P$ satisfying $\int_\R \P \ud z =1$, there exists a unique $\lambda \in \R$ such that:
\begin{equation} \label{eq:deflambda2}
\int_\R \s(\lambda,z) \P(z) \ud z = 1.
\end{equation}
This defines a function: $\P \mapsto \lambda(\P)\in \R$.\\

This claim is a consequence of Lemma \ref{lem:Laplace} applied to the function $a \mapsto \beta(a) \int_\R  s(a,z) \P(z) \ud z$, because Assumption \eqref{hyp:th2} implies:
\[
1= \int_\R \P(z) \ud z
\le
\int_\R  \int_{0}^\infty  e^{-\lambda_0 a} \beta(a) s(a,z) \P(z) \ud z \ud a
\le 
\int_{0}^\infty  e^{-\lambda_0 a} \beta(a) \ud a < \infty,
\]
using $s(a,z)\le 1$. We get the uniform lower bound $\lambda(P) \ge \lambda_0$.

\item We define the mixing operator:
\begin{equation} \label{eq:def-M2}
\M(\P) = \G\left( \s(\lambda(\P),\cdot) \,  \P \right).
\end{equation}
\end{enumerate}

As previously, the Schauder theorem on $\mathcal C$, defined by \eqref{eq:defC}, is applied to conclude.

\begin{itemize}
\item The set $\C$ is preserved by $\M$. The last estimate, is a consequence of Lemmas \ref{lem:moments}, \ref{lem:Var_sphi} with $\nu = \s(\lambda,\cdot)$ and $W(z)=z^2$:
\[
\var \M(\P)  = \var G + \frac12 \var \Big( \s(\lambda(\P),\cdot) \,  \P \Big) \le \var G + \frac12 \var  \P \le 2 \var G.
\]

\item The set $\M(\mathcal{C})$ is relatively compact thanks to the Fréchet–Kolmogorov theorem.
The estimate on $\norm{\tau_h (\M \P) - \M \P}_{\L^1(\R)}$ is easier to obtain:
\begin{eqnarray*}
\norm{\tau_h (\M \P) - \M \P}_{\L^1(\R)} &\le &  \norm{\tau_h G -G}_{\L^1(\R)} \norm{\s(\lambda(\P),\cdot) \P}_{\L^1(\R)} \norm{\s(\lambda(\P),\cdot) \P}_{\L^1(\R)}\\
& = & \norm{\tau_h G -G}_{\L^1(\R)} \to 0.
\end{eqnarray*}

\item The operator $\M:\mathcal{C}\to\mathcal{C}$ is continuous.
Let $\P_1, \P_2 \in \C$, and let us denote: $\lambda_k = \lambda(\P_k)$, $\s_k=\s(\lambda_k,\cdot)$, for $k=1,2$.
Assumption \eqref{hyp:th2} implies $\lambda_k \ge \lambda_0$ and $\norm{\s_k}_{\L^\infty(\R)} \le \norm{\s(\lambda_0,\cdot)}_{\L^\infty(\R)} \le \int_0^\infty e^{-\lambda_0 a} \beta(a) \ud a< \infty$.

\medskip

Thanks to \eqref{eq:def-Gconv} and \eqref{eq:a2b2} we have:
\begin{eqnarray*}
&& \norm{\M(\P_1)-\M(\P_2)}_{\L^1(\R)} \\
&\le& \norm{G}_{\L^1(\R)} \norm{(\s_1 \P_1) * (\s_1 \P_1) - (\s_2 \P_2) * (\s_2 \P_2)}_{\L^1(\R)} \\
&=& \norm{\left[ (\s_1 - \s_2) \P_1 + \s_2 (\P_1-\P_2) \right] * \left(\s_1 \P_1 + \s_2 \P_2\right) }_{\L^1(\R)} \\
&\le& \left(\norm{\s_1 - \s_2}_{\L^\infty(\R)} \norm{\P_1}_{\L^1(\R)} + \norm{\s_2}_{\L^\infty(\R)} \norm{\P_1-\P_2}_{\L^1(\R)} \right)
2 \int_0^\infty e^{-\lambda_0 a} \beta(a) \ud a
\end{eqnarray*}

It remains to prove that $\L^1(\R) \ni P \mapsto \s \in \L^\infty(\R)$ is continuous.
Let us fix $\P_1$, define $\lambda = \lambda_2-\lambda_1$, and:  $\h(x) = \abs{1 - e^{-x}}$.

\medskip

As previously, we have:
\begin{equation} \label{eq:sP1}
\abs{\int_\R (\s_1-\s_2) \P_1 \ud z} = \abs{ \int_\R \s_2 (\P_1-\P_2) \ud z} \le \norm{\P_1-\P_2}_{\L^1(\R)} \int_0^\infty e^{-\lambda_0 a} \beta(a) \ud a.
\end{equation}

Let us define the function $\H(\lambda)$ by:
\begin{eqnarray} \label{eq:sP2}
\H(\lambda) &:=& \int_\R \int_0^\infty \h(\lambda a) \, e^{-\lambda_1 a} \P_1(z) s(a,z) \beta(a) \ud a \ud z  \nonumber\\
&=& \abs{ \int_\R \int_0^\infty \left(1 - e^{-(\lambda_2-\lambda_1) a}\right) e^{-\lambda_1 a} \P_1(z) s(a,z) \beta(a) \ud a \ud z} \nonumber\\
&=& \abs{\int_\R (\s_1-\s_2) \P_1 \ud z}. 
\end{eqnarray}
The equalities hold because $a\mapsto e^{-\lambda_1 a} - e^{-\lambda_2 a}$ has a constant sign.

It is easily checked that the function $\H$ enjoys the following properties of $\h$: $\H(0)=0$, it is continuous, decreasing for $\lambda \le 0$ and increasing for $\lambda \ge 0$. In particular: as $\H(\lambda)$ goes to $0$, $\abs{\lambda}$ tends to $0$.

As $\norm{\P_1-\P_2}_{\L^1(\R)}$ goes to $0$ we get that $\abs{\lambda}=\abs{\lambda_2-\lambda_1}$ tends to $0$ thanks to \eqref{eq:sP1}, \eqref{eq:sP2}. Lastly:
\[
\norm{\s_1 - \s_2}_{\L^\infty(\R)} \le \int_0^\infty \abs{e^{-\lambda_1 a} - e^{-\lambda_2 a}} \beta(a) \ud a
\]
tends to $0$ as $\abs{\lambda_1-\lambda_2}$ goes to $0$. 
\end{itemize}

\end{proof}

\section{Conclusions and perspectives}

As the Kre\u{\i}n-Rutman theory could not be applied to Problems \eqref{eq:pb1}, \eqref{eq:pb2}, because the operator $\G$ is nor linear nor monotone,
we used the Schauder fixed point theorem to prove the existence of principal eigenelements. Another difficulty was that the dispersion relation was not explicit.
Using regularizing effects of the operator $\G$ and properties of moments we were able to build some convenient convex subset of $\L^1(\R)$ on which the Schauder theorem was applied, to prove the joint existence of a principal eigenvector and of the corresponding principal eigenvalue.

\medskip

The previous approach can be straightforwardly adapted to the case of a genetic component of offspring traits $z$ being a convex combination of the parental traits $z^\prime,  z^{\sec}$: $z=\alpha z^\prime + \beta z^{\sec}$ with $\alpha + \beta = 1$ and $0<\alpha,\beta<1$. The resulting operator $\G$ is defined as:
\begin{eqnarray*}\label{eq:def-Gbis}
\G f(z) &=& \dfrac1{\int_\R f(x)\, dx} \, \iint_{\R^2} f(z^\prime) f(z^\sec) \, G\left(z - (\alpha z^\prime + \beta z^{\sec})\right) \ud z^\prime \ud z^\sec,\\
        &=& \frac{1}{\alpha \beta} \dfrac1{\int_\R f(x)\, \ud x} \, \left(G * f\left(\frac{\cdot}{\alpha}\right) * f\left(\frac{\cdot}{\beta}\right)\right)(z).
\end{eqnarray*}

This operator conserves the 0th order and the 1st order moments, and the main confinement property remains true: for $\phi \in \L^1(\R)$ satisfying $\int_\R \phi(z) \ud z =1$ and $\int_\R z \phi(z) \ud z = 0$, we have:
\[
\var\G(\phi) = \var G + (\alpha^2 + \beta^2) \var \phi,
\]
with $\alpha^2+\beta^2 < 1$.

\bigskip

It would be interesting to include more general birth rates than the constant ones in our analysis. Weither the principal eigenelements are unique or not is also an interesting question. As concerns the asymptotic behaviour of Equations \eqref{eq:noage}, \eqref{eq:age}, convergence to natural attractors could be investigated: $e^{-\lambda t} f(t,z) \to F(z)$, resp. $e^{-\lambda t} f(t,a,z) \to F(a,z)$, in some weighted norm, where $F(z),F(a,z)$ are defined by \eqref{eq:pb1}, \eqref{eq:pb2}. 






\bigskip

In the context of a changing environment, see \cite{CR} for instance, we can assume that the optimal trait moves at a speed $c$,
by changing the mortality rate $\mu(z)$ by $\mu(z-ct)$ in Equations \eqref{eq:noage}, \eqref{eq:age}.
For instance, for the homogeneous model, in the moving frame $\tilde z = z - ct$, Problem \eqref{eq:pb1} writes:
\[
\begin{cases}
\ds  \lambda F(z) - c \, \p_z F(z) + \mu(z) F(z) = \G \left(F \right)(z), \medskip\\
\ds \int_\R F(z) \ud z = 1.
\end{cases}
\]

For this problem, the function $F$ cannot be supposed to be symmetric.
Thus our confinement proof cannot be applied to this case. 
Other models could also be  studied: cyclic changes of environment, space structured models with migration, see \cite{CR}.

\vspace*{\fill}

\textbf{Acknowledgments.} The authors thank Dr. \'E. Bouin for fruitful discussions. T.B's, V.C.'s, J.G.'s research is supported by the European Research Council (ERC) under the European Union's Horizon 2020 research and innovation programme, ERC Starting Grant MESOPROBIO No. 639638.

\vspace*{\fill}

\newpage

\bibliographystyle{abbrv}
\bibliography{biblio}

\end{document}